\documentclass[letterpaper, 10 pt, conference]{ieeeconf}  
\usepackage{gen_settings}

\IEEEoverridecommandlockouts                              

\title{\LARGE \bf
Safety-Critical Controller Verification via Sim2Real Gap Quantification
}

\author{Prithvi Akella, Wyatt Ubellacker, and Aaron D. Ames$^{1}$
\thanks{This work was supported by the AFOSR Test and Evaluation Program, grant FA9550-19-1-0302 and Dow (\#227027AT).}
\thanks{$^{1}$Authors are with the California Institute of Technology (email
        {\tt\small \{pakella, wubellac, ames\}@caltech.edu})}%
}

\begin{document}

\maketitle
\thispagestyle{empty}
\pagestyle{empty}

\begin{abstract}

The well-known quote from George Box states that: ``All models are wrong, but some are useful." To develop more useful models, we quantify the inaccuracy with which a given model represents a system of interest, so that we may leverage this quantity to facilitate controller synthesis and verification.  Specifically, we develop a procedure that identifies a sim2real gap that holds with a minimum probability.  Augmenting the nominal model with our identified sim2real gap produces an uncertain model which we prove is an accurate representor of system behavior.  We leverage this uncertain model to synthesize and verify a controller in simulation using a probabilistic verification approach.  This pipeline produces controllers with an arbitrarily high probability of realizing desired safe behavior on system hardware without requiring hardware testing except for those required for sim2real gap identification.  We also showcase our procedure working on two hardware platforms - the Robotarium and a quadruped.

\end{abstract}


\section{Introduction}
The nominal controller synthesis process for safety-critical systems follows a well-worn path: develop a model for the system of interest (from first principles, system identification, or otherwise), develop a controller in simulation based on this model, implement the controller on the safety-critical system of interest, and, most likely, tune controller parameters until the system exhibits the desired behavior.  In what will follow, we offer a method that aims to mitigate the need for extensive - and perhaps expensive - hardware testing and verification, by simultaneously verifying the simulator used for controller synthesis and the controller itself.  The hope is that such a method will augment existing model generation techniques to streamline the controller synthesis and verification procedure.

Our desire to augment existing model generation techniques stems from a desire to leverage extensive prior work in system identification and reduced-order-model control.  System identification specifically deals with techniques aimed at generating models that more closely align with the system-to-be-modeled~\cite{tangirala2018principles,morelli2016aircraft,ljung1998system,keesman2011system}. This has led to the development of state-of-the-art methods for this process, \textit{e.g.} the Volterra method which extends linear convolution to nonlinear systems~\cite{schetzen2006volterra,koh1985second} and the NARMAX method which represents system evolution as a nonlinear transform of prior states, inputs, and recorded outputs~\cite{billings2013nonlinear,rodriguez1997genetic,chen2007narx}.  There have also been more recent efforts at representing systems via Neural Networks~\cite{nelles2020nonlinear,chen1990non,billings2005new} or regressing dynamical models via Gaussian Process Regression~\cite{kocijan2005dynamic,bernardo1998regression,rodriguez2021learning}.  Additionally, regardless of the method used, there are techniques to progressively make better models that more closely align with their corresponding systems.

\begin{figure}[t]
    \centering
    \includegraphics[width = \columnwidth]{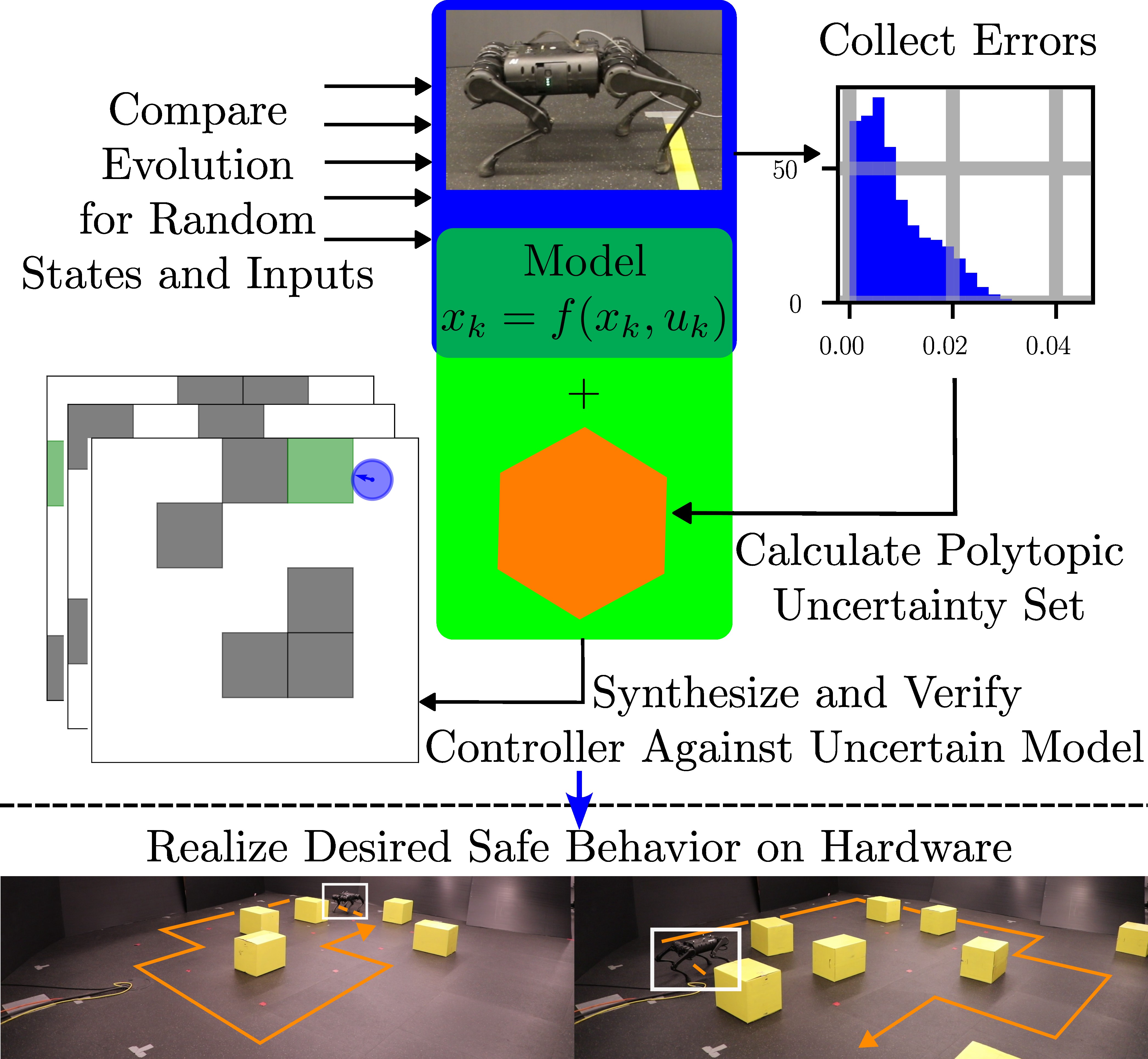}
    \caption{An overview of our proposed approach.  Models are often inaccurate despite best practices to generate accurate models.  So, we propose a method to probabilistically determine an uncertainty set to augment our nominal model.  Doing so produces an uncertain model that we prove accurately represents true system behavior to an arbitrarily high probability.  This allows for simultaneous controller synthesis and verification in simulation while retaining safe performance with limited testing on hardware.  (Quadruped highlighted in white in bottom figures)}
    \label{fig:title}
    \vspace{-0.2 in}
\end{figure}

Despite the capability of existing methods to generate good models, they oftentimes fail to capture rarer system behavior, \textit{i.e.} corner cases, stochastic disturbances, \textit{etc}.  Accounting for such behavior underlies the study into stochastic nonlinear models for system identification~\cite{abdalmoaty2019identification,giri2010block}.  While very expressive, these models tend to be quite computationally complex.  Moreover, prior work indicates that reduced order models are oftentimes sufficient to express underlying system physics~\cite{garofalo2012walking,thieffry2018control,xiong2018bipedal}.  These simpler models are also often leveraged in autonomy stacks for complex systems.  Additionally, augmenting these reduced order models with uncertainty bounds and accounting for these bounds through robust control~\cite{thieffry2018control,su1999reduced}, input-to-state stabilizing barrier functions~\cite{kolathaya2018input,taylor2020learning,xu2015robustness}, or other techniques, \textit{i.e.} learning~\cite{mesbah2018stochastic,soloperto2018learning,schneider1996exploiting}, tends to yield safe and reliable controllers.  This prompts the question then: how should the ``correct" uncertainty bound be determined?  If we use too large a bound, the controller might be too conservative, and if we use too small a bound, then the controller could be unsafe.

\newidea{Our Contribution:} Our work aims to address this question of determining the magnitude of uncertainty one should consider for a provided model that will be used in controller synthesis.  In this vein, our results are two-fold.  First, we provide a norm bound on the uncertainty between a provided model and its corresponding system.  This bound effectively reads as: with minimum probability $1-\epsilon$, the evolution of the system at one time-step will lie within the bounding region prescribed by the evolution of the provided model plus an uncertainty lying within a polytopic set we calculate.  Our second result makes use of this uncertainty bound and the authors' prior work in stochastic verification~\cite{akella2022scenario,akella2022sample} to verify the generated controller against the uncertain model.  This results in a pipeline for safety-critical controller synthesis and verification that translates to hardware performance.

\newidea{Implementation:} We showcase our contributions on two hardware platforms, the Georgia Tech Robotarium~\cite{wilson2020robotarium} and a quadruped.  For both platforms, we leverage our first result to determine their discrepancies with respect to a nominal unicycle model.  Then, we use this discrepancy to synthesize and verify a controller, purely in simulation, that is guaranteed to successfully steer the uncertain unicycle model in a navigation and static obstacle avoidance scenario.  For the Robotarium, we also add multiple, uncontrolled, stochastically evolving agents.  For both systems, we then show that these verified controllers exhibit similar levels of performance on hardware without any parameter tuning or testing required - effectively working "out-of-the-box" after our procedure.  Furthermore, these controllers successfully steered their systems to satisfy their control objectives despite a wide variety of randomized test scenarios.

\newidea{Paper Structure:} First, we provide the necessary mathematical background for our procedure in Section~\ref{sec:backgroundinfo} and formally define the sim2real gap in Section~\ref{sec:prob_statements}.  We describe our approach to sim2real gap quantification in Section~\ref{sec:sim2real_quantification} and our controller verification procedure in Section~\ref{sec:controller_generation_verification}.  Finally, we detail all experimental demonstrations of our procedure in Section~\ref{sec:experimentation}.

\section{A Brief Review of Scenario Optimization}
\label{sec:backgroundinfo}
In short, both aspects of our procedure arise from separate uses of the scenario optimization procedure~\cite{campi2008exact,campi2018wait}.  Scenario optimization identifies robust solutions to uncertain convex optimization problems of the following form:
\begin{equation}
    \label{eq:uncertain_program}
    \tag{UP}
    \begin{aligned}
        \sovar^* & = \argmin_{\sovar \in \sospace \subset \mathbb{R}^d}~ & &c^T\sovar, \\
        &~~\mathrm{subject~to}~ & &\sovar \in \sospace_{\delta},~\delta \in \Delta.
    \end{aligned}
\end{equation}
Here, \eqref{eq:uncertain_program} is the uncertain program as $\delta \in \Delta$ is a sample of some random variable in the probability space $\Sigma = (\Omega, \mathcal{F}, \prob)$, with sample space $\Omega = \Delta$, (perhaps) unknown event space $\mathcal{F}$, and (perhaps) unknown probability measure $\prob$.  Convexity is assured via assumed convexity in the spaces $\sospace$ and $\sospace_{\delta}$.  Furthermore, since $\Delta$ is typically a set of infinite cardinality, \textit{i.e.} $|\Delta| = \infty$, identification of a solution $\sovar^*$ such that $\sovar^* \in \sospace_{\delta}~\forall~\delta \in \Delta$ is hard.

To resolve this issue, the study of scenario optimization solves a related optimization problem formed from an $N$-sized sample set of the random constraint samples $\delta$ and provides a probabilistic guarantee on the robustness of the corresponding solution $\sovar^*_N$.  Specifically, if we were to take an $N$-sized set of samples $\deltaset$, we could construct the following scenario program:
\begin{equation}
    \label{eq:scenario_program}   
    \tag{RP-N}
    \begin{aligned}
        \sovar^*_N & = \argmin_{\sovar \in \sospace \subset \mathbb{R}^d}~ & &c^T\sovar, \\
        &~~\mathrm{subject~to}~ & &\sovar \in \bigcap\limits_{\delta_i \in \{\delta_j\}_{j=1}^N} \sospace_{\delta_i}.
    \end{aligned}
\end{equation}
\noindent Then, we require the following assumption.
\begin{assumption}
\label{assump:RPN_solvability}
The scenario program~\eqref{eq:scenario_program} is solvable for any $N$-sample set $\deltaset$ and has a unique solution $\sovar^*_N$.
\end{assumption}

Assumption~\ref{assump:RPN_solvability} then guarantees existence of a scenario solution $\sovar^*_N$ for~\eqref{eq:scenario_program} for any provided sample set $\deltaset$.  As such, we can define a set containing those samples $\delta \in \Delta$ to which the scenario solution $\sovar^*_N$ is not robust, \textit{i.e.} $
F(\sovar) = \{\delta \in \Delta~|~\sovar \not \in \sospace_{\delta}\}. $
With this set definition we can formally define the \textit{violation probability} of our solution.
\begin{definition}
\label{def:violation}
The \textit{violation probability} $V(\sovar)$ of a given $\sovar \in \sospace$ is defined as the probability of sampling a constraint $\delta$ to which $\sovar$ is not robust, \textit{i.e.} $V(\sovar) = \prob[\delta \in F(\sovar)]$ .
\end{definition}
\noindent Then, the main theorem is as follows:
\begin{theorem}[Adapted from Theorem 1 in~\cite{campi2008exact}]
\label{thm:scenario_opt}
Let Assumption~\ref{assump:RPN_solvability} hold.  The following inequality is true:
\begin{equation}
    \prob^N[V(\sovar^*_N) > \epsilon] \leq \sum_{i=0}^{d-1} \binom{N}{i} \epsilon^i(1-\epsilon)^{N-i}.
\end{equation}
\end{theorem}
\noindent In Theorem~\ref{thm:scenario_opt} above, $N$ is the number of sampled constraints $\delta$ for the scenario program~\eqref{eq:scenario_program}; $\sovar^*_N$ is the scenario solution to the corresponding scenario program; $V(\sovar^*_N)$ is the violation probability of that solution as per Definition~\ref{def:violation}; $d$ is the dimension in which $\sovar$ lies, \textit{i.e. } $\sovar \in \mathbb{R}^d$; and $\prob^N$ is the induced probability measure over sets of $N$-samples of $\delta$ given the probability measure $\prob$ for $\delta$.  
\section{Probabilistic Sim2Real Gap Quantification}
Our method for sim2real gap quantification will express the identification of such a gap as a convex optimization problem with (perhaps) infinite constraints, and we will take a scenario approach to solve this problem.  This approach will yield a robust result - a large enough sim2real gap - that holds with some minimum probability.  The next section will formally define what we mean by a sim2real gap.
\subsection{Defining the Gap}
\label{sec:prob_statements}
First, we denote our true system via $x$ and our nominal model via $\hat x$, \textit{i.e.} $\forall~k,j=0,1,2,\dots$,
\begin{equation}
    \tag{SYS}
    \label{eq:systems}
    \begin{aligned}
        & \mathrm{\textbf{True:}}~& x_{k+1} & = f(x_k,u_k),~x_{k,k+1} \in \mathcal{X},~u_k \in \mathcal{U}, \\
        & \mathrm{\textbf{Sim:}}~ & \hat x_{j+1} & = \hat f(\hat x_j,\hat u_j),~\hat x_{j,j+1} \in \hat{\mathcal{X}},~\hat u_{j} \in \hat{\mathcal{U}}.
    \end{aligned}
\end{equation}
As an example consistent with the demonstrations to follow, the true system could be a quadruped with the nominal model a unicycle system we aim to represent it with.

To provide a method of comparing the evolution of the two systems in~\eqref{eq:systems}, we will define two maps - $M_x$ which projects the true system state $x$ to the model state $\hat x$, and $M_u$ which extends the model input $\hat u$ to the true system input $u$:
\begin{equation}
    \tag{MAPS}
    \label{eq:maps}
        M_x: \mathcal{X} \to \hat{\mathcal{X}}, \quad M_u: \hat{\mathcal{U}} \times \mathcal{X} \to \mathcal{U}.
\end{equation}
These maps in~\eqref{eq:maps} let us formalize the gap we aim to identify between the two systems.  First, we assume that we can command an input $u$ to the true system by prescribing an input $\hat u$ that we would provide to our associated model.  Then, we assume we can measure the projected true system state $M_x(x_k)$ at some time-step $K$, \textit{i.e.}, $\forall~k=0,1,\dots,K$,
\begin{equation}
    \tag{OBS}
    \label{eq:observation_function}
   x_{k+1} = f(x_k,M_u(\hat{u},x_k)),~ O(x_0,\hat{u}) = M_x(x_K).
\end{equation}

To put this in the context of the quadruped/unicycle model example, the underlying control loop operates at $1$ kHz and provides a natural discrete abstraction at that time-step.  Since, we desire our unicycle model to update at $10$ Hz then, $K=100$.  $M_x$ is just the projection of the quadruped state to its unicycle components.  Likewise, $M_u$ is the underlying control loop that runs at $1$ kHz to realize the commanded forward walking speed and rotation.  While these maps seem abstract, we will provide examples in Section~\ref{sec:experimentation}.

This observation map $O$ in~\eqref{eq:observation_function} permits us to quantify a discrepancy between model evolution and observed true system evolution.  However, we only want to make this comparison when the projection of the initial state $M_x(x_0) \in \hat{\mathcal{X}}$ - as otherwise, the projected initial state is not addressed by our representative model.  This results in the following space definition and problem statement:
\begin{equation}
    \label{eq:desired_true_ss}
    \Pi(M_x) = \left\{x \in \mathcal{X}~|~M_x(x) \in \hat{\mathcal{X}}\right\}
\end{equation}

\begin{definition}
    \label{def:sim2real_gap}
    Let $O$ be as defined in~\eqref{eq:observation_function}, $\hat f, \mathcal{X}, \hat{\mathcal{U}}$ be as defined in~\eqref{eq:systems}, and $\Pi(M_x)$ be as defined in~\eqref{eq:desired_true_ss}.  The \textit{sim2real gap} $\simrealgap \in \mathbb{R}$ is such that $~\forall~(x_0,\hat u) \in \Pi(M_x) \times \hat{\mathcal{U}}$,
    \begin{equation}
        \label{eq:sim2real_formal_eq}
        \simrealgap \geq \left\|O(x_0,\hat{u}) - \hat f\left(M_x(x_0), \hat{u}\right) \right\|.
\end{equation}
\end{definition}

\subsection{Quantifying the Gap}
\label{sec:sim2real_quantification}
To start, we express identification of the sim2real gap $\simrealgap$ in~\eqref{eq:sim2real_formal_eq} as an optimization problem:
\begin{align}
    \label{eq:sim2real_base}
        \simrealgap & = \argmin_{r \in\mathbb{R}}~ & &r, \\
        &~~\mathrm{subject~to}~ & &r \geq  \left\|O(x_0,\hat{u}) - \hat f\left(M_x(x_0), \hat{u}\right) \right\|, \\
        & & & \qquad \qquad \dots \forall~(x_0,\hat{u}) \in \Pi(M_x) \times \hat{\mathcal{U}}
\end{align}

This problem is (likely) impossible to solve as posed.  So, taking inspiration from prior work, we aim instead to find a ``good" solution to~\eqref{eq:sim2real_base} via a randomized, sample-based approach (See Section 5 in~\cite{akella2022sample}).  Such a solution will only hold with some minimum probability $\epsilon \in [0,1)$.  Although, we can make $\epsilon$ arbitrarily close to $1$ with enough samples.

\newidea{Description of our Approach:} Our approach hinges on the ability to independently draw state and input samples from a static distribution $\pi$ over the combined state and input spaces $\mathcal{X} \times \hat{\mathcal{U}}$.  In the experimental demonstrations to follow, we will argue and show evidence that our chosen method produces independent samples from an unknown distribution $\pi$.  However, the specifics of crafting such a distribution will likely be different for different system/model pairs - this is where we anticipate a large portion of future work in this vein to lie.  So, for the moment, we will simply assume the existence of such a distribution $\pi$ formalized as follows:

\begin{definition}
\label{def:comparison_distribution}
The \textit{comparison distribution} $\pi$ maps subsets of the combined state and model input space $\mathcal{X} \times \hat{\mathcal{U}}$ - as defined in~\eqref{eq:systems} - to $[0,1]$ \textit{i.e.} $\forall~A \subseteq \mathcal{X} \times \hat{\mathcal{U}},~\pi(A) \in [0,1]$.  Furthermore, $\pi$ ``covers" the space of states and inputs generating constraints for~\eqref{eq:sim2real_base}, \textit{i.e.}, $\pi \left( \Pi(M_x) \times \hat{\mathcal{U}} \right) = 1$.
\end{definition}

Using this comparison distribution $\pi$, we can construct a scenario program for sets of $N$ samples $\left\{\left(x^l_0,\hat{u}^l\right)\right\}_{l=1}^N$ of state and input pairs $(x_0,\hat{u})$ distributed by $\pi$:
\begin{align}
    \label{eq:sim2real_scenario}
        \simrealgap^*_N & = \argmin_{r \in\mathbb{R}}~ & &r, \\
        &~~\mathrm{subject~to}~ & &r \geq  \left\|O(x^j_0,\hat{u}) - \hat f\left(M_x(x_0), \hat{u}^j\right) \right\|, \\
        & & & \quad \dots \forall~(x^j_0,\hat{u}^j) \in \left\{\left(x^l_0,\hat{u}^l\right)\right\}_{l=1}^N.
\end{align}
The resulting solution $\simrealgap^*_N$ is an ever-increasing lower bound on $\simrealgap$ as expressed in the following theorem.
\begin{theorem}
\label{thm:sim2real_theorem}
Let $\simrealgap^*_N$ be the solution to~\eqref{eq:sim2real_scenario} with $O$ as defined in~\eqref{eq:observation_function}, $\hat f$ as defined in~\eqref{eq:systems}, and $\pi$ as defined in Definition~\ref{def:comparison_distribution}.  Then, $\forall~\epsilon \in [0,1]$
\begin{equation}
    \begin{gathered}
    S_1 \triangleq \prob_{\pi}\left[\simrealgap^*_N \geq \left\|O(x_0,\hat{u}) - \hat f\left(M_x(x_0), \hat{u}\right) \right\| \right], \\
    \prob^N_{\pi}\left[S_1 \geq 1-\epsilon \right] \geq 1-(1-\epsilon)^N.
    \end{gathered}
\end{equation}
In other words, $\simrealgap^*_N$ is larger than the sim2real gap for any sampled state and input pair $(x_0,\hat{u})$ from $\pi$ with minimum probability $1-\epsilon$ and confidence $1-(1-\epsilon)^N$.
\end{theorem}
\begin{proof}
This proof follows almost directly from Theorem~\ref{thm:scenario_opt}, noting that the violation probability,
\begin{equation}
    V(\simrealgap^*_N) = 1-S_1.
\end{equation}
Then, reversing inequalities and changing arguments appropriately provides the desired result.
\end{proof}

\begin{figure}[t]
    \centering
    \includegraphics[width = \columnwidth]{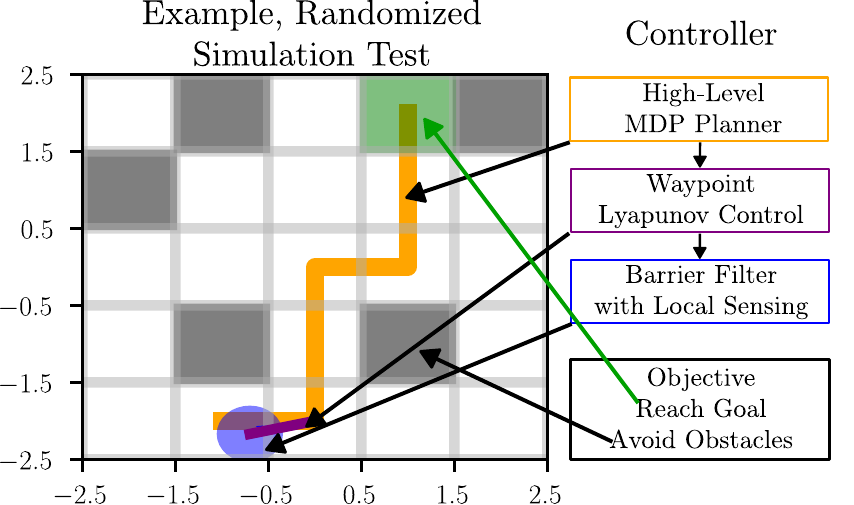}
    \caption{As part of our pipeline, we verify controllers against our uncertain model.  Shown above is an example, randomized test scenario for controller development for the quadruped.   The test scheme remains the same for the two different systems (Robotarium and quadruped) as their controller objectives are similar.  This information is further explained in Section~\ref{sec:experimentation}.}
    \vspace{-0.2 in}
    \label{fig:controller_architecture}
\end{figure}

\noindent In other words, Theorem~\ref{thm:sim2real_theorem} states that the solution $\simrealgap^*_N$ to~\eqref{eq:sim2real_scenario} is a decent approximation of the sim2real gap $\simrealgap$.

\section{Safety-Critical Controller Verification}
\label{sec:controller_generation_verification}
Now, we can leverage this approximate sim2real gap $\simrealgap^*_N$ to facilitate controller synthesis and verification in simulation. This section details those efforts.

\newidea{Constructing a Valid Uncertain Model:} To start, we can use the resulting probabilistic sim2real gap $\simrealgap^*_N$ from~\eqref{eq:sim2real_scenario} to augment our nominal model $\hat x$ in~\eqref{eq:systems} and define an uncertain system denoted via $\Tilde x$.  Specifically, we will first define a feasible space of disturbances,
\begin{equation}
    \label{eq:uncertain_polytope}
    D = \{d \in \mathbb{R}^n~|~\|d\| \leq \simrealgap^*_N\},
\end{equation}
and use $D$ to define our uncertain system:
\begin{equation}
    \label{eq:uncertain_model}
\Tilde x_{j+1} = \hat f(\Tilde x_j,\hat u_j) + d_j,~d_j \sim \uniform[D],
\end{equation}
with $\uniform[D]$ the uniform distribution over $D$.  If we define a one-step reachable space provided a model state and input $(\hat x, \hat u) \in \hat{\mathcal{X}} \times \hat{\mathcal{U}}$,
\begin{equation}
    \label{eq:uncertain_reachable_space}
    \mathcal{R}(\hat x, \hat u) = \left\{\hat f(\hat x,\hat u) + d,~\forall~d \in D\right\}, 
\end{equation}
then we have the following result regarding the evolution of the true system and this reachable space.
\begin{corollary}
\label{corr:model_accuracy}
Let $O$ be as defined in~\eqref{eq:observation_function}, $\mathcal{R}$ be as defined in~\eqref{eq:uncertain_reachable_space}, $M_x$ be as defined in~\eqref{eq:maps}, $\pi$ be as per Definition~\ref{def:comparison_distribution}, and $\simrealgap^*_N$ be as defined in~\eqref{eq:sim2real_scenario}.  Then, $\forall~\epsilon \in[0,1]$,
\begin{equation}
    \begin{gathered}
    S_2 \triangleq \prob_{\pi}\left[O(x_0,\hat{u}) \in \mathcal{R}(M_x(x_0),\hat{u}) \right], \\
    \prob^N_{\pi}\left[S_2 \geq 1-\epsilon \right] \geq 1-(1-\epsilon)^N.
    \end{gathered}
\end{equation}
In other words, the probability that the observed evolution of the true system $O(x_0,\hat{u})$ lies in the reachable space of our uncertain model $\mathcal{R}(M_x(x_0),\hat{u})$ is at least $1-\epsilon$ with confidence $1-(1-\epsilon)^N$.
\end{corollary}
\begin{proof}
This is a direct application of Theorem~\ref{thm:sim2real_theorem}, as for any state and input pair $(x_0,\hat{u}) \in \mathcal{X} \times \hat{\mathcal{U}}$, $\simrealgap^*_N \geq \left\|O(x_0,\hat{u}) - \hat f\left(M_x(x_0), \hat{u}\right) \right\|$ iff $O(x_0,\hat{u}) \in \mathcal{R}(M_x(x_0),\hat{u})$, by definition of $\mathcal{R}$ in~\eqref{eq:uncertain_reachable_space} and $D$ in~\eqref{eq:uncertain_polytope}.
\end{proof}

\begin{figure}[t]
    \centering
    \includegraphics[width = \columnwidth]{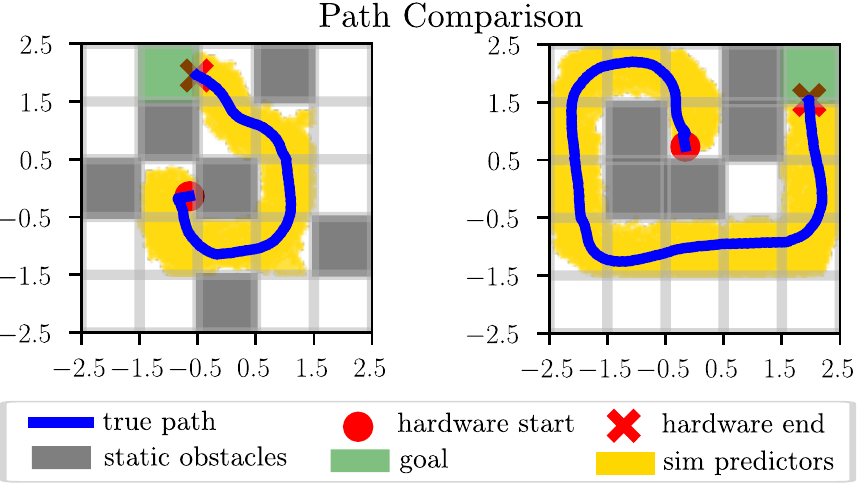}
    \vspace{-0.2 in}
    \caption{Our procedure increases the confidence that those controllers that pass the verification step will exhibit similar performance on hardware as they did in simulation, even if we did not directly verify the controller on hardware.  This increased confidence arises through our verification of the uncertain model, whose reachable set we prove encapsulates true system evolution to high probability.  This can be seen in the figures above, as the quadruped's evolution (blue) lies within its associated uncertain simulator's predictions (gold).  This information is further explained in Section~\ref{sec:experimentation}.}
    \label{fig:robotarium_synthesis}
    \vspace{-0.2 in}
\end{figure}

\noindent In other words, Corollary~\ref{corr:model_accuracy} tells us that even though a single step of our uncertain model~\eqref{eq:uncertain_model} may not be an accurate representation of our true system's evolution, the space of all possible evolutions does, to high probability, contain the evolution of our true system.

For controller synthesis and verification then, Corollary~\ref{corr:model_accuracy} tells us that our uncertain model is a decent approximator of true system behavior.  Therefore, any controller that exhibits good performance on the uncertain model, should likewise exhibit good performance on the true system.  Quantification of ``good" performance and verifying a controller's ability to realize ``good" performance is the subject of the next subsection, which follows from the risk-aware probabilistic verification procedure detailed in Section 3 in~\cite{akella2022scenario}.

\begin{figure*}[t]
    \centering
    \includegraphics[width = \textwidth]{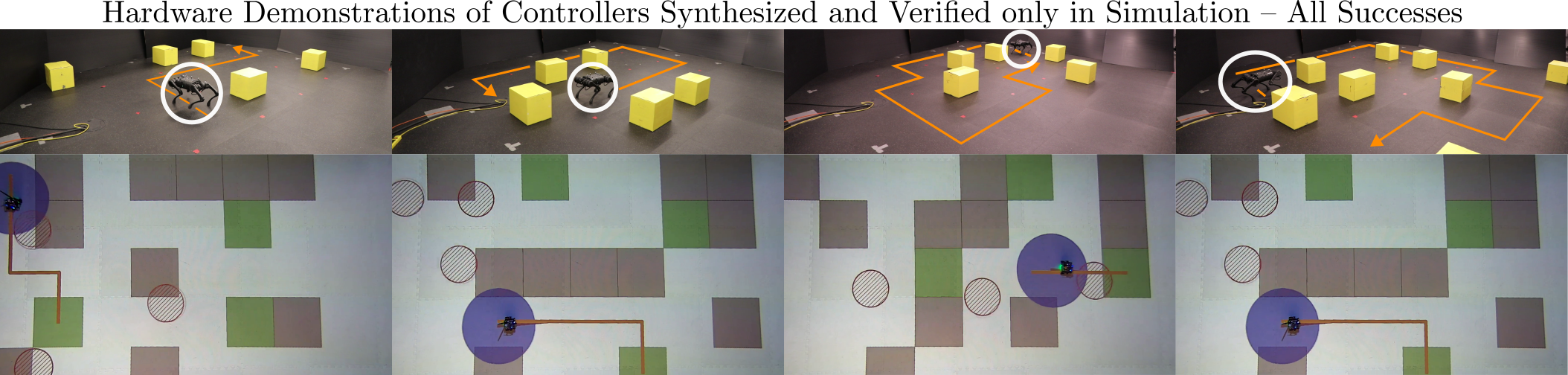}
    \caption{Since we verified our controllers against the uncertain model produced by our procedure, we expect that the closed-loop hardware systems should realize similar, satisfactory behavior.  Indeed, for the first $10$ runs on the quadruped and the first $40$ runs on the Robotarium, the agents were able to avoid static/moving obstacles and navigate to their goals successfully, despite a wide variety of randomized test scenarios.  The first four runs for both systems are depicted above.  This ability to synthesize and verify controllers in simulation, with confidence that similar behaviors will manifest in the true system without requiring additional testing, is the main benefit of our proposed approach.  Paths for all tests are shown in orange, and the quadruped is highlighted in white.  The multi-level control architecture is depicted in Figure~\ref{fig:controller_architecture}.}
    \label{fig:experiment_tiles}
    \vspace{-0.2 in}
\end{figure*}

\newidea{Verifying against Safety Metrics:}  First, provided a parameterized controller $\modelcontroller: \hat{\mathcal{X}} \times \Theta \to \hat{\mathcal{U}}$ with parameter $\theta \in \Theta$ and noise sequence $\xi$ where $\xi_j = d \sim \uniform[D]$, we define $\phi$ to be the closed-loop trajectory to our uncertain model~\eqref{eq:uncertain_model}, \textit{i.e.} with $J > 0$ and $\Tilde x_0 = \hat x_0$,
\begin{equation}
    \label{eq:uncertain_trajectory}
    \phi^{\modelcontroller}(\hat x_0,\xi,\theta,J) = \Tilde x_J,~\Tilde x_{j+1} = \hat f\left(\Tilde x_j,\modelcontroller(\Tilde x_j,\theta)\right) + \xi_j.
\end{equation}
Second, inspired by traditional safety measures, \textit{e.g.} barrier functions over the system state, we define safety metrics $\safetymetric$ to be functions over system trajectories that only output positive numbers for trajectories exhibiting the desired safe behavior:
\begin{equation}
    \label{eq:robustness_metric}
    \safetymetric\left(\phi^{\modelcontroller}\left(\hat x_0, \xi, \theta \right) \right) \geq 0 \iff
    \begin{gathered}
        \phi^{\modelcontroller}\left(\hat x_0, \xi, \theta \right) \mathrm{~exhibits} \\
        \mathrm{desired~safe~behavior}.
    \end{gathered}
\end{equation}
Examples of safety metrics include robustness measures from Signal Temporal Logic~\cite{donze2010robust} or the minimum value of a barrier function over a finite-time horizon~\cite{ames2016control}.  

Ideally, we would like for our controller $\modelcontroller$ to only ever realize trajectories $\phi^{\modelcontroller}(\hat x_0,\xi,\theta)$ with a positive evaluation under this safety metric $\safetymetric$.  To check for this positivity, we will first draw $(\hat x_0,\theta)$ uniformly from $\hat{\mathcal{X}} \times \Theta$.  Then, we will evaluate the safety of one trajectory emanating from that initial condition $\hat x_0$ with that parameter $\theta$, \textit{i.e.} record $s = \safetymetric\left(\phi^{\hat{\mathcal{U}}}\left(\hat x_0, \xi, \theta \right) \right)$ for some valid noise-sequence $\xi$.  Repeating this procedure $N$ times to take $N$ such samples $s_i$ and create the dataset $\{s_i\}_{i=1}^N$, we then define $s^*_N = \min\{s_i\}_{i=1}^N$.  Then, via Corollary 2 in~\cite{akella2022scenario}, we have the following result:
\begin{corollary}
\label{corr:randomized_verification}
Let the uncertain system trajectory $\phi^{\modelcontroller}(\hat x_0,\xi,\theta)$ be as defined in~\eqref{eq:uncertain_trajectory}, the safety metric $\safetymetric$ be as defined in~\eqref{eq:robustness_metric}, $\{s_i\}_{i=1}^N$ be the safety values of $N$ sampled trajectories $\phi^{\modelcontroller}(\hat x_0,\xi,\theta)$ with initial conditions and parameters $(x_0,\theta)$ drawn from $\uniform[\hat{\mathcal{X}}\times\Theta]$, and $s^*_N = \min\{s_i\}_{i=1}^N$.  Then, $\forall~\epsilon\in[0,1]$ and abbreviating $\uniform[\hat{\mathcal{X}}\times\Theta] = \pi_0$,
\begin{equation}
    \begin{gathered}
    S_3 \triangleq \prob_{\pi_0,~\xi_j \sim \uniform[D]~\forall~j=1,2,\dots}\left[ s \geq s^*_N \right], \\
    \prob^N_{\pi_0,~\xi_j \sim \uniform[D]~\forall~j=1,2,\dots}\left[S_3 \geq 1-\epsilon \right] \geq 1-(1-\epsilon)^N.
    \end{gathered}
\end{equation}
In other words, the probability that $s^*_N$ will be smaller than any sample-able safety value $s$ is at minimum $1-\epsilon$ with confidence $1-(1-\epsilon)^N$.
\end{corollary}
\begin{proof}
    This is an application of Corollary 2 in~\cite{akella2022scenario}.
\end{proof}

\begin{figure*}[t]
    \centering
    \includegraphics[width = \textwidth]{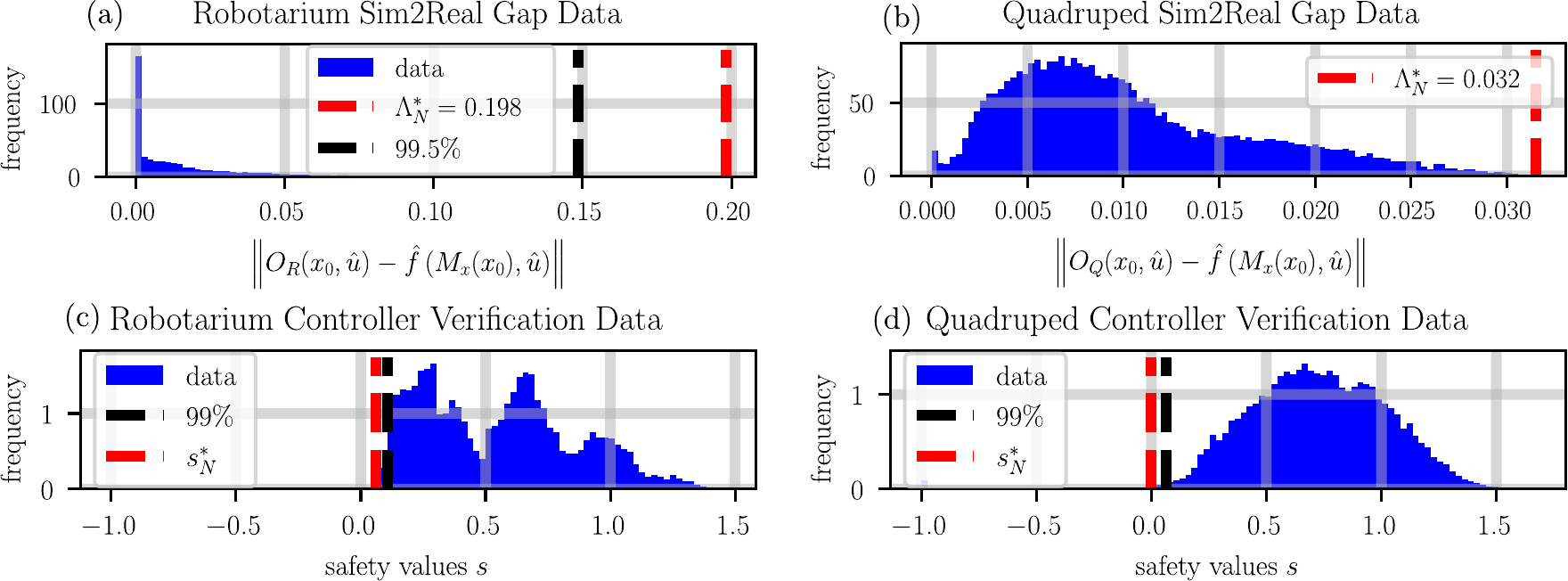}
    \caption{All data for our experimental pipeline.  (Top Left) Data for calculation of a probabilistic sim2real gap $\simrealgap^*_N$ for the Robotarium.  Per Theorem~\ref{thm:sim2real_theorem}, we expect that after observing $600$ randomly sampled errors, our reported sim2real gap $\simrealgap^*_N = 0.198$ is greater than any sample-able gap with minimum probability $99.5\%$ with $95\%$ confidence.  Comparing $\simrealgap^*_N$ to the true cutoff after taking $2400$ samples verifies this inequality and supports Theorem~\ref{thm:sim2real_theorem}.  (Top Right) Data for sim2real gap calculation for the quadruped.  True cutoffs are not shown as we did not exhaustively sample gaps.
    (Bottom) Verification data for both controllers against their respective uncertain models.  Note that in both cases, we sampled $300$ trajectories to calculate a minimum safety value $s^*_N$ which, according to Corollary~\ref{corr:randomized_verification}, should be less than any sample-able safety value with minimum probability $99\%$ with $95\%$ confidence.  Taking $20000$ safety samples and calculating the true cutoffs against the sampled data shows that this inequality holds verifying Corollary~\ref{corr:randomized_verification}.}
    \label{fig:experiment_data}
    \vspace{-0.2 in}
\end{figure*}

This completes the theoretical statement of our proposed pipeline for safety-critical controller verification via sim2real gap quantification.  For the intermediate synthesis step, one can really use any method they like \textit{e.g.} robust control methods~\cite{zhou1998essentials,green2012linear}, input-to-state-stable Lyapunov or barrier functions~\cite{xu2015robustness,kolathaya2018input,sontag1995characterizations,hespanha2008lyapunov}, \textit{etc}.  The emphasis here is on verifying the resulting controller against the uncertain model and using Corollary~\ref{corr:randomized_verification} to discriminate between better controllers - those with a higher minimum probability of realizing safe trajectories - and worse controllers - those with a lower minimum probability.  By Theorem~\ref{thm:sim2real_theorem} and Corollary~\ref{corr:randomized_verification} we know that our uncertain model is a decent representor of system behavior.  Therefore, those controllers with a high probability of exhibiting desired safe behavior in the uncertain simulator should likely have a high probability of exhibiting desired safe behavior on hardware.  We will demonstrate this procedure on two hardware platforms - the Robotarium~\cite{wilson2020robotarium} and a quadruped.
\section{Experimental Demonstrations}
\label{sec:experimentation}
We will describe our procedure's implementation on both systems simultaneously, as we aim to represent both systems with the same model abstraction, a unicycle:
\begin{equation}
    \label{eq:experimental_systems}
    \begin{aligned}
        x_{k+1} & = f(x_k,u_k),~x_{k,k+1} \in \mathcal{X},~u_k \in \mathcal{U},~t_{k+1} - t_k = \Delta t \\
        \hat x_{j+1} & = \hat x_j + \Delta \hat{t}
        \begin{bmatrix}
        \cos\left(\hat x[3]\right) & 0 \\
        \sin\left(\hat x[3]\right) & 0 \\
        0 & 1
        \end{bmatrix}\hat u_j
        ,~\hat x_{j,j+1} \in \hat{\mathcal{X}},~\hat u_{j} \in \hat{\mathcal{U}}.
    \end{aligned}
\end{equation}
The specific parameters for each system are as follows (Robotarium (R) and Quadruped (Q)):
\begin{enumerate}
    \item[(R)] $\hat{\mathcal{X}} = [-1.6 \times 1.6] \times [-1,1] \times [0,2\pi]$, $\hat{\mathcal{U}} = [-0.2,0.2] \times [-\pi,\pi]$, $\Delta t,\Delta \hat{t} = 0.033$.
    \item[(Q)] $\hat{\mathcal{X}} = [-2.5,2.5]^2 \times [0,2\pi]$, $\hat{\mathcal{U}} = [-0.15,0.15] \times [-0.3,0.3]$, $\Delta t = 0.001$, and $\Delta \hat{t} = 0.1$.
\end{enumerate}

For both systems, we can read true state data $x$ to recover the idealized unicycle state $\hat x$.  Therefore, $M_x$ is just a projection for both systems.  The Robotarium permits unicycle-like commands to their agents resulting in $M_u = I_{2 \times 2}$.  On the other hand, the quadruped has a lower-level walking controller that operates at $1$ kHz to realize commanded forward walking and yaw angular velocities \cite{ubellacker2023icra}.  As a result, $M_u$ for the quadruped is this pre-built walking controller.  Finally, our observation maps:
\begin{enumerate}
    \item[(R)] $O_R(x_0,\hat{u}) = M_x(x_1)$ - read the projected true state after one time-step, and,
    \item[(Q)] $O_Q(x_0,\hat{u}) = M_x(x_{100})$ - read the projected true state after $100$ time-steps.
\end{enumerate}

\newidea{Sampling from the Comparison Distribution:} To calculate the discrepancy between the system and our chosen model, we will sample from the comparison distribution $\pi$ (Definition~\ref{def:comparison_distribution}) with the steps listed below:
\begin{enumerate}
\item[$(x_0)$] We uniformly randomly sample a planar position $x$ from $\hat{\mathcal{X}}$.  Then, we send both agents to the waypoint $x$ using a Lyapunov controller built on top of the input maps $M_u$ for both systems. Once the system reaches a ball of $0.1$ m around the desired waypoint, we stop and record the resulting state as $x_0$.
\item[$(\hat{u})$] We uniformly randomly sample an input $\hat {u}$ from $\hat{\mathcal{U}}$.  Then, we command the system with this input for $50$ true-system time-steps for the Robotarium and $1000$ true-system time-steps for the quadruped.  We command this input for an extended period to approximate the randomized initial location that we had before sampling $x_0$ so that subsequent samples drawn from this procedure are drawn effectively independently.
\end{enumerate}

\newidea{Probabilistic Sim2Real Gap Calculation:}  For the Robotarium, we collected $2400$ observations and used the first $600$ to calculate constraints for~\eqref{eq:sim2real_scenario}.  This provided a sim2real gap constant $\simrealgap^*_N = 0.198$ which, according to Theorem~\ref{thm:sim2real_theorem}, is greater than any sampled sim2real gap with minimum probability $99.5\%$ with minimum confidence $95\%$.  Figure~\ref{fig:experiment_data}~(a) shows the probabilistic sim2real gap $\simrealgap^*_N$ overlaid on a histogram of sampled sim2real gaps to approximate the underlying distribution.  Notice that since $\simrealgap^*_N$ is indeed greater than the $99.5\%$ cutoff, this corroborates Theorem~\ref{thm:sim2real_theorem}.  Figure~\ref{fig:experiment_data}~(b) shows similar results for the quadruped after taking $100$ observations.  The true cutoff value is not shown as we did not exhaustively sample observations to approximate the underlying distribution.  Although, with Theorem~\ref{thm:sim2real_theorem} and the Robotarium results, we are confident that the calculated sim2real gap is greater than any sample-able gap with minimum probability $97\%$ and with confidence $95\%$.

\newidea{Safety-Critical Controller Verification:} Figure~\ref{fig:controller_architecture} shows the general controller architecture for which we aim to identify parameters such that the resulting controller $\hat{U}$ has a high probability of rendering satisfactory behavior on the uncertain model.  To synthesize such a controller, we use a safety metric $\safetymetric$ as per~\eqref{eq:robustness_metric} that outputs $-1$ if the agent crashes into either a static or moving obstacle and outputs the Manhattan distance traveled along the shortest feasible path to a goal - the orange line in Figure~\ref{fig:controller_architecture} - if it successfully avoids crashes within $200$ time-steps.  We do not formally define this metric as it is not central to the paper's concept.

For both systems, we iterated through at least $10$ different sets of controllers $\modelcontroller$ with parameter spaces, $\Theta =$
\begin{enumerate}
    \item[(R)] All possible setups of $10$ static obstacles, $3$ goals, and $1$ initial condition cell in an $8\times5$ grid such that a feasible path exists between the initial cell and at least one goal.  This is in addition to the starting locations and movement directions of $3$, uncontrolled moving obstacles.
    \item[(Q)] All possible setups of $5$ static obstacles, $1$ goal, and $1$ initial condition cell in a $5\times5$ grid such that a feasible path exists between the initial cell and the goal.
\end{enumerate}
Once we found controllers that, according to Corollary~\ref{corr:randomized_verification}, had a minimum satisfaction probability of at least $99\%$ with $95\%$ confidence, we implemented these controllers $\modelcontroller_R,\modelcontroller_Q$ on their respective systems, the Robotarium and the quadruped.  In each verification step, we evaluated the controllers under $500$ randomized scenarios - $500$ test scenarios we would have otherwise had to run on real systems were we not using our uncertain model to approximate true-system behavior.  Figures~\ref{fig:experiment_data}~(c) and (d) show the verification data for the finalized controllers $\modelcontroller_R,\modelcontroller_Q$, respectively, with the distribution approximated by evaluating $20000$ randomized scenarios.

According to our pipeline then, for the fact that we verified our controllers - Corollary~\ref{corr:randomized_verification} - against an uncertain model which has a high probability of representing true system behavior - Theorem~\ref{thm:sim2real_theorem} and Corollary~\ref{corr:model_accuracy} - these controllers should similarly exhibit satisfactory behavior on their true systems when implemented.  Figure~\ref{fig:experiment_tiles} depicts the first four tests underwent by both systems - they successfully completed their tasks in these tests as expected.  We ran $40$ more randomized tests for the Robotarium, drawing from the same parameter space $\Theta$ as described earlier - all successes.  Similarly, we ran $10$ more tests for the Quadruped drawing from its respective parameter space $\Theta$ - all successes.  We expected this level of performance since the controllers exhibited a high probability of realizing desired safe behaviors on their respective uncertain models that encapsulated true system behavior.  Furthermore, we did not have to run any tests on hardware to gain this level of confidence in our controller, except for those tests required to calculate the sim2real gap.

\section{Conclusion}
We present a pipeline for safety-critical controller verification via sim2real gap quantification.  Our pipeline starts by augmenting the nominal model with an uncertainty set generated via a probabilistic sim2real gap analysis between the model and the system it represents.  By using this uncertain model for synthesis and verification, we limit the number of hardware tests we have to run to develop effective, safe controllers that exhibit satisfactory performance on hardware.  We showcase our procedure successfully developing satisfactory controllers on two hardware platforms - the Robotarium and a quadruped.  In future work, we hope to add a randomized synthesis step to our pipeline, to automate the generation of safe, effective controllers while minimizing the need for extensive/expensive hardware tests throughout.

\bibliographystyle{IEEEtran}
\balance
\bibliography{IEEEabrv,bib_works}

\end{document}